\documentclass[a4paper, USenglish, cleveref, autoref, thm-restate, numberwithinsect]{lipics-v2021}
\title{When do homomorphism counts help in query algorithms?}
\author{Balder ten Cate}{University of Amsterdam}{b.d.tencate@uva.nl}{https://orcid.org/0000-0002-2538-5846}{Supported by the European Union’s Horizon 2020 research and innovation programme (MSCA-101031081).}
\author{Victor Dalmau}{Universitat Pompeu Fabra}{victor.dalmau@upf.edu}{https://orcid.org/0000-0002-9365-7372}{Supported by the MiCin under grants PID2019-109137GB-C22 and PID2022-138506NB-C22, and the Maria de Maeztu program (CEX2021-001195-M).}
\author{Phokion G.\ Kolaitis}{University of California Santa Cruz and IBM Research }{kolaitis@ucsc.edu}{https://orcid.org/0000-0002-8407-8563}{Partially supported by NSF Grant IIS-1814152.}
\author{Wei-Lin Wu}{University of California Santa Cruz}{wwu53@ucsc.edu}{https://orcid.org/0009-0004-3341-1508}{}
\authorrunning{ten Cate, Dalmau, Kolaitis, Wu}
\Copyright{Balder ten Cate, Victor Dalmau, Phokion G.~Kolaitis, and Wei-Lin Wu}
\category{}
\relatedversion{}
\nolinenumbers

\usepackage{colonequals}
\usepackage{amsmath}
\usepackage{ifthen}

\newcommand{\mathmode}[1]{\begin{math}#1\end{math}}
\newcommand{\defas}{\colonequals} 
\newcommand{\etc}{\cdots} 
\newcommand{\etl}{\ldots} 
\newcommand{\suchthat}{:} 
\newcommand{\set}{\mathcal} 
\newcommand{\sete}[1]{\{#1\}} 
\newcommand{\setm}[2]{\{#1 \suchthat #2\}} 
\newcommand{\sett}[2]{\{#1 \suchthat \text{#2}\}} 
\newcommand{\size}[1]{\left|#1\right|} 
\newcommand{\union}{\mathbin{\cup}} 
\newcommand{\intsc}{\mathbin{\cap}} 
\newcommand{\bunion}{\mathop{\bigcup}} 
\newcommand{\cartpwr}[2]{#1^{#2}} 
\newcommand{\bool}{\mathbb{B}} 
\newcommand{\nat}{\mathbb{N}} 
\newcommand{\adom}{\mathrm{adom}} 
\newcommand{\dsum}{\mathbin{\oplus}} 
\newcommand{\dprod}{\mathbin{\otimes}} 
\newcommand{\inc}{\mathrm{inc}} 
\newcommand{\block}{\mathrm{block}} 
\newcommand{\homto}{\mathrel{\rightarrow}} 
\newcommand{\homeq}{\mathrel{\leftrightarrow}} 
\newcommand{\homtype}[1]{\left[#1\right]_{\homeq}} 
\newcommand{\sur}{\mathrm{sur}} 
\newcommand{\class}{\mathcal} 
\newcommand{\satis}{\models} 
\newcommand{\blor}{\bigvee} 
\newcommand{\bland}{\bigwedge} 
\newcommand{\modclass}{\mathrm{Mod}} 
\newcommand{\canq}[1]{\mathit{q}^{#1}} 
\newcommand{\csp}{\mathrm{CSP}} 

\theoremstyle{definition}
\newtheorem{question}[theorem]{Question}

\theoremstyle{definition}
\newtheorem{Open Problem}[theorem]{Open Problem}

\EventEditors{Graham Cormode and Michael Shekelyan}
\EventNoEds{2}
\EventLongTitle{27th International Conference on Database Theory (ICDT 2024)}
\EventShortTitle{ICDT 2024}
\EventAcronym{ICDT}
\EventYear{2024}
\EventDate{March 25--28, 2024}
\EventLocation{Paestum, Italy}
\EventLogo{}
\SeriesVolume{290}
\ArticleNo{8}

\begin{document}
\maketitle
\begin{abstract}
A query algorithm based on homomorphism counts is a procedure for determining whether a given instance satisfies a property  by counting homomorphisms between the given instance and finitely many predetermined instances. In a left query algorithm, we count homomorphisms from the predetermined instances to the given instance, while in a right query algorithm we count homomorphisms from the given instance to the predetermined instances. 
Homomorphisms are usually counted over the semiring $\nat$ of non-negative integers; it is also meaningful, however, to count homomorphisms over the Boolean semiring $\bool$, 
 in which case the homomorphism count  indicates whether or not a homomorphism exists. 
 We first characterize the properties that admit a left query algorithm over $\bool$ by showing that these are precisely the properties that are both   first-order definable and closed under homomorphic equivalence.
 After this, we turn attention to a comparison between left query algorithms over $\bool$ and left query algorithms over $\nat$. In general, there are properties that admit a left query algorithm over $\nat$ but not over $\bool$. The main result of this paper asserts that if a property is closed under homomorphic equivalence, then that property admits a left query algorithm over $\bool$ if and only if it admits a left query algorithm over $\nat$. In other words and rather surprisingly, homomorphism counts over $\nat$ do \emph{not} help as regards properties that are closed under homomorphic equivalence. Finally, we characterize the properties that admit both a left query algorithm over $\bool$ and a right query algorithm over $\bool$.
\end{abstract}

\section{Introduction}\label{sec:intro}
Consider a scenario in which we 
are interested in a certain property  of database instances and we wish to find out whether or not a given instance $A$ satisfies that property by asking finitely many predetermined queries  against $A$.  Naturally, which 
properties can be checked in this way depends on what
kind of queries we are allowed to ask. For example,
if we are restricted to using Boolean conjunctive queries and  evaluating them
under set semantics, then  only 
properties that are invariant under homomorphic
equivalence stand a chance to be checked in this way. In particular, this means that we cannot
test in this way whether a  relation in a  given instance contains precisely 5 tuples.
In contrast, if we are also allowed to ask for the
number of homomorphisms from a given conjunctive query
to $A$ (a feature that is supported by actual  database management systems), then 
we can find out much more about the instance $A$, including whether
it satisfies 
the aforementioned property.

In this paper, we embark on a systematic study of this scenario. As usual, by a \emph{property} of  instances we mean a class $\mathcal C$ of  instances closed under isomorphisms. We write $\bool$ for the Boolean semiring with $\lor$ and $\land$ as its operations, while we write $\nat$ for the semiring of the non-negative integers with $+$ and $\times$ as its operations.
Formally, we say that a 
class $\class{C}$ of instances admits a  \emph{left query algorithm over $\bool$}  if there exists a finite set $\class{F}=\{F_1, \ldots, F_k\}$ 
of instances such that the membership in $\class{C}$  of an arbitrary
instance $A$  is completely determined
by $\hom_\bool(\class{F},A)$, where $\hom_\bool(\class{F},A)$
is the $k$-tuple of Boolean values indicating for each $i\leq k$
whether or not there is a homomorphism $F_i\to A$. Similarly,
we say that a 
class $\class{C}$ of instances admits a \emph{left query algorithm over $\nat$} if there exists a finite
set  $\class{F}=\{F_1, \ldots, F_k\}$ of instances   such that the membership in $\class{C}$  of an arbitrary
instance $A$ is completely determined
by $\hom_\nat(\class{F},A)$, where $\hom_\nat(\class{F},A)$
is the $k$-tuple of non-negative integers indicating for each $i\leq k$
how many homomorphisms from $F_i$ to $A$ there are. \emph{Right query algorithms}
over $\bool$ and over $\nat$ are defined
 similarly, except that we now count the homomorphisms from $A$ to each $F_i$ instead of the
homomorphisms from each $F_i$ to $A$.

Assume that the class $\class{C}$ of instances under consideration  admits a left query algorithm over $\bool$ or over $\nat$ using the set $\class{F}=\{F_1, \ldots, F_k\}$. Let $q^{F_i}$ be the canonical conjunctive query associated with the instance $F_i$, $1\leq i\leq k$.
Then $\hom_\bool(\class{F},A)$ is the $k$-tuple of Boolean values denoting whether  $q^{F_i}(A)=1$ or $q^{F_i}(A)=0$, i.e., whether  $q^{F_i}$ is  true   or false on $A$ under set semantics. Similarly, 
$\hom_\nat(\class{F},A)$ is the $k$-tuple of non-negative integers that are the answers to $q^{F_i}$ on $A$ under bag-set semantics. Thus, intuitively,
a class $\class{C}$ admits a left query algorithm over
$\bool$ or over $\nat$ if membership in $\class{C}$ is answerable by evaluating finitely many Boolean conjunctive queries under set semantics or under bag-set semantics, respectively. Observe  that  these are data complexity notions because the queries $q^{F_1}, \ldots,q^{F_k}$ are fixed while the instance $A$ varies. Observe also that these two notions differ in expressive power: for example, if $\class{C}$ is the class of instances in which a particular relation $R$ has precisely 5 tuples, then $\class{C}$ admits a left query algorithm over $\nat$, but not over $\bool$.

There are two pieces of earlier work (each with different motivation and results) where the notion of a left query algorithm or variants of this notion have been explored.  First, Bielecki and Van den Bussche \cite{DBLP:conf/icdt/BieleckiB03}  defined what it means for a query $p$ to be \emph{derivable} through interrogation with a query language $L$ using a \emph{database independent strategy},
where the interrogation consists of asking  the cardinality $|q(A)|$ for finitely many queries $q\in L$. When $L$ is the language of conjunctive queries with no existential quantifiers, such strategies correspond to
left-query algorithms over $\mathbb N$;
whereas, when $L$ is the language of 
Boolean conjunctive queries, they correspond to left-query algorithms over $\mathbb B$. Second, when the instances are unordered graphs,
the concept of a left query algorithm over $\nat$ was
studied by
Chen et al.~\cite{chen2022algorithms} under the name 
\emph{non-adaptive query algorithm}; note that Chen et al.~\cite{chen2022algorithms} were apparently unaware of the work by Bielecki and Van den Bussche \cite{DBLP:conf/icdt/BieleckiB03}.
The term ``non-adaptive'' is apt as it conveys that  the instances $F_i$ (or the associated conjunctive queries $q^{F_i}$), $1\leq i \leq k$, in the set $\mathcal F$ depend only on the class $\class{C}$ and do not change during a run of the query algorithm. It is also natural to consider \emph{adaptive query algorithms},
where the instances $F_i$, $1\leq i \leq k$, are not required
to be fixed  up front. In fact, such adaptive notions were explored in both
\cite{DBLP:conf/icdt/BieleckiB03} and \cite{chen2022algorithms}. In particular, 
Chen et al.~\cite{chen2022algorithms}  showed that  adaptive left query algorithms over $\nat$ are more powerful than  non-adaptive ones over $\nat$. 
It is easy to see, however, that the existence of an adaptive left query algorithm over $\bool$ implies the existence of a 
non-adaptive one over $\bool$. In this sense, adaptive left query algorithms over $\bool$ are not more powerful than non-adaptive ones over $\bool$.  Also, since in this paper we study non-adaptive algorithms only (but over both $\bool$ and $\nat$), we will not use the adjective ``non-adaptive'' here.  

Our investigation begins by 
focusing on left query algorithms over $\bool$. 
It is easy to see that a class of instances admits a left query algorithm over $\bool$ if and only if it is definable by a Boolean combination of conjunctive queries.
Using tools developed by Rossman \cite{rossman2008homomorphism}
to prove the preservation-under-homomorphisms theorem in the finite, we obtain a deeper  characterization of such classes  by showing that a class of instances admits a left query algorithm over $\bool$
if and only if it is both first-order definable and closed under homomorphic equivalence. Clearly, if a class of instances is closed under homomorphic equivalence, then it is a (possibly infinite) union of homomorphic equivalence classes.
 We show that if $\class{C}$ is a finite union of homomorphic-equivalence classes, then $\class{C}$ admits a left query algorithm over $\bool$ if and only if every homomorphic-equivalence class in that union admits a left query algorithm over $\bool$. In contrast, a similar result does not hold for arbitrary infinite unions.

After this, we turn attention to a comparison between left query algorithms over $\bool$ and left query algorithms over $\nat$. As discussed earlier, left query algorithms over $\nat$ are more powerful than left query algorithms over $\bool$. The intuitive reason is that left query algorithms over $\bool$ do not differentiate between homomorphically equivalent instances, while those over $\nat$ do. 
The main (and technically more challenging) result of this paper reveals that, in a precise sense, this is the \emph{only} reason why left query algorithms over $\nat$ are more powerful than left query algorithms over $\bool$.
More precisely, our main theorem
asserts that if a class $\class{C}$ is closed under homomorphic equivalence, then 
$\class{C}$
admits a left query algorithm over $\bool$ if and only if $\class{C}$ admits a left query algorithm over $\nat$. In other words and rather surprisingly, homomorphism counts over $\nat$ do \emph{not} help as regards properties that are closed under homomorphic equivalence. 
As an immediate consequence, a constraint satisfaction problem  has a left query algorithm over $\nat$ if and only if this problem is first-order definable.

Finally, we characterize the properties that admit both a left query algorithm over $\bool$ and a right query algorithm over $\bool$. In particular, we show that a class $\class{C}$ of instances admits both a left query algorithm over $\bool$ and a right query algorithm over $\bool$ if and only if $\class{C}$ is definable by a Boolean combination of Berge-acyclic conjunctive queries.
To see the point of this result, recall that if a class admits a left query algorithm over $\bool$, then it is definable by a Boolean combination of conjunctive queries. Thus, if the class admits also a right query algorithm over $\bool$, then these conjunctive queries  can be taken to be Berge-acyclic.

\smallskip

\noindent{\bf Related work}~ 
A classical result by
Lov\'asz~\cite{lovasz1967operations} characterizes graph isomorphism in terms of ``left'' homomorphism counts: two graphs $G$ and $H$ are isomorphic if and only if for every graph $F$, the number of homomorphisms from $F$ to $G$ is equal to the number of homomorphisms from $F$ to $H$. In more recent years, there has been a study of relaxations of isomorphisms obtained by requiring that the number of homomorphisms from $F$ to $G$ is equal to the number of homomorphisms from $F$ to $H$, where $F$ ranges over a restricted class of graphs~\cite{DBLP:journals/jgt/Dvorak10,DBLP:conf/icalp/DellGR18,DBLP:conf/mfcs/BokerCGR19}. Furthermore, a study of  relaxations of isomorphism obtained by counting the number of ``right'' homomorphisms to a restricted class was carried out in \cite{atserias2021expressive}.

There has been an extensive body of research on answering 
queries under various types of access restrictions.
Closer in spirit to the work reported here is   view determinacy, which is the question of
whether the
answers to a query can be inferred when given access
only to a certain view of the database instance~\cite{Nash2010}. 
We note that  view determinacy is typically concerned with non-Boolean queries and non-Boolean views; in contrast, 
the question of whether a class admits a left query algorithm over $\bool$ can be interpreted as the question of whether there is a finite set of Boolean conjunctive queries that determine a given Boolean query.  A study of view determinacy under bag-set semantics was recently initiated in~\cite{Kwiecien2022:determinacy}. 
Section~\ref{sec:conclus} contains additional commentary on the relationship between left query algorithms over $\nat$ and view determinacy under bag-set semantics.

\section{Basic Notions}\label{sec:prelim}
\noindent{\bf Relational database instances}~
A \emph{relational database schema} or, simply, a \emph{schema} is a finite set $\sigma = \sete{R_1, \etl, R_m}$ of relation symbols $R_i$, each of which has  a positive integer $r_i$ as its associated \emph{arity}. A \emph{relational database instance} or, simply, an \emph{instance} is a tuple $A = (R_1^{A}, \etl, R_m^{A})$, where each $R_i^{A}$  is  a relation of arity $r_i$.  The \emph{facts} of the instance $A$ are the tuples in the relations $R_i^{A}$, $1\leq i \leq m$. 
The \emph{active domain of $A$}, denoted $\adom(A)$, is the set of all entries occurring in the facts of $A$. All instances $A$ considered are assumed to be finite, i.e., $\adom(A)$ is finite.
A \emph{graph} is an instance $A$ over a schema consisting of a binary relation symbol $E$ and such that $E^A$ is a binary relation that is symmetric and irreflexive.

The \emph{incidence multigraph} $\inc(A)$ of  an instance $A$ is the  bipartite multigraph whose  parts  are the sets $\adom(A)$ and $\block(A)= \sett{(R, t)}{\mathmode{R \in \sigma} and \mathmode{t \in R^{A}}}$, and whose edges are the pairs $(a,(R,t))$ such that $a$ is one of the entries of $t$. A \emph{path of length $n$}  in $A$  is a sequence  $a_0, a_1, \etl, a_n$ of elements in $\adom(A)$ for which there are elements $b_1, \etl, b_n$ in $\block(A)$ such that the sequence $a_0, b_1, a_1, \etl, b_n, a_n$ is a path in $\inc(A)$ in the standard graph-theoretic sense (disallowing traversing an  edge twice in succession in opposite directions). Two elements $a$ and $a'$ in $\adom(A)$ are \emph{connected} if  $a = a'$ or  there is a path $a_0, a_1, \etl, a_n$ in $A$ with $a = a_0$ and $a' = a_n$. We say that $A$ is \emph{connected} if every two elements $a$ and $a'$ in $\adom(A)$ are connected. A \emph{cycle of length $n$} in $A$ is a path of length $n$ in $A$  with  $a_n = a_0$. We say $A$ is \emph{acyclic} if it contains no cycles. The \emph{girth of $A$} is   the shortest length of a cycle in $A$ or $\infty$ if $A$ is acyclic.
  \looseness=-1

An instance  $A$ is a \emph{subinstance} of an instance $B$ if
$R^A \subseteq R^B$, for every $R \in \sigma$.

A \emph{class} $\mathcal C$ of instances is a collection  of instances over the same schema that is closed under isomorphism (i.e., if $A \in \mathcal{C}$ and $B$ is isomorphic to $A$, then $B\in \mathcal{C}$). Every decision problem $P$ about instances can be identified with the class of all ``yes'' instances of $P$.

\smallskip

\noindent{\bf Homomorphisms, conjunctive queries, and canonical instances}
A \emph{homomorphism} from an instance $A$ to an instance $B$ is a function  $h: \adom(A) \to \adom(B)$ such that for every relation symbol $R \in \sigma$ with arity $r$ and for all 
  elements $a_1, \etl, a_r$ in $\adom(A)$ with $(a_1, \etl, a_r) \in R^A$,
 we have  $(h(a_1), \etl, h(a_r)) \in R^B$.  We write $h:A \rightarrow B$ to denote that $h$ is a homomorphism from $A$ to $B$; we also write $A\rightarrow B$ to denote that there is a homomorphism from $A$ to $B$. We say that $A$ and $B$ are \emph{homomorphically equivalent}, denoted  $A \homeq B$,  if  $A \homto B$ and $B \homto A$. Clearly, $\homeq$ is an equivalence relation on instances. We write  $\homtype{A}$ to denote the equivalence class of $A$ with respect to $\homeq$, i.e., $\homtype{A} = \setm{B}{B \homeq A}$.

Let $\mathcal C$ be a class of instances. We say that $\mathcal C$ is \emph{closed under homomorphic equivalence} if whenever $A\in {\mathcal C}$ and $A\homeq B$, we have that $B\in {\mathcal C}$. As an example, for every instance $A$, we have that the equivalence class $\homtype{A}$ is closed under homomorphic equivalence. For a different example, the class of all $3$-colorable graphs is closed under homomorphic equivalence.

We assume familiarity with the syntax and the semantics of first-order logic (FO). For a FO-sentence $\varphi$, we denote
by $\modclass(\varphi)$ the set $\setm{A}{A \satis \varphi}$ of instances $A$ that satisfy $\varphi$ under the active domain semantics (i.e., the quantifiers range over elements of the active domain of the instance at hand). A \emph{Boolean conjunctive query} (Boolean CQ) is  a FO-sentence of the form $\exists x_1 \etl x_n (\alpha_1 \land \etc \land \alpha_k)$, where each $\alpha_j$ is an atomic formula of the form $R(y_1, \etl, y_r)$, each variable $y_i$ is among the variables $x_1,\ldots, x_n$, and each variable $x_i$ occurs in at least one of the atomic formulas $\alpha_1,\ldots,\alpha_k$.

The \emph{canonical instance} of a conjunctive query $q$, denoted $A^q$, is the instance whose active domain consists of the variables of $q$, and whose facts are the conjuncts of $q$. Conversely, the \emph{canonical conjunctive query} of an instance $A$, denoted $q^A$, has, for each $a$ in $\adom(A)$, an existentially quantified variable $x_a$ and, for each fact $(a_1, \etl, a_r) \in R^{A}$, a conjunct $R(x_{a_1}, \ldots, x_{a_r})$. An immediate consequence of the semantics of FO is that, for every two instances $A$ and $D$,
we have that $D\models q^A$ if and only if $A\rightarrow D$.

A conjunctive query is 
\emph{Berge-acyclic} if its canonical
instance is acyclic, as defined earlier.   The notion of Berge-acyclicity is stronger than the more standard notion of acyclicity in databse theory,  which requires that the conjunctive query has a join tree (see, e.g., \cite{abiteboul1995foundations}).

\smallskip

\noindent{\bf Homomorphism counts, left and right profiles} Let $\nat=(N,+,\times, 0,1)$ be the semiring of the non-negative integers and let $\bool=(\{0,1\}, \lor, \land, 0,1)$ be the Boolean semiring.
If $A$ and $B$ are two instances, then we write  $\hom_\nat(A, B)$ for the number of homomorphisms from $A$ to $B$. 
For example, if $A$ is a graph and $K_3$ is the $3$-clique, then $\hom_\nat(A,K_3)$ is the number of $3$-colorings of $A$. 
We extend this notion to $\hom_\bool(A, B)$, where $\hom_\bool(A, B) = 1$ if there is a homomorphism from $A$ to $B$, and $\hom_\bool(A, B) = 0$ otherwise. For example, $\hom_\bool(A,K_3)=1$ if $A$ is $3$-colorable, and $\hom_\bool(A,K_3)=0$ if $A$ is not $3$-colorable.

Let $\set{F} = \sete{F_1, \etl, F_k}$ be a finite non-empty set of instances and let $A$ be an instance.
\begin{itemize}
    \item 
 The \emph{left profile of $A$ in $\set{F}$ over $\nat$} is the tuple
    
  \centerline{  $\hom_\nat(\set{F}, A) = (\hom_\nat(F_1, A), \etl, \hom_\nat(F_k, A)).$}
\item The  \emph{left profile of $A$ in $\set{F}$ over $\bool$} is the tuple
    
    \centerline{$\hom_\bool(\set{F}, A) = (\hom_\bool(F_1, A), \etl, \hom_\bool(F_k, A))$.}
\item The \emph{right profile of $A$ in $\set{F}$ over $\nat$} is the tuple

\centerline{$\hom_\nat(A, \set{F}) = (\hom_\nat(A, F_1), \etl, \hom_\nat(A, F_k))$.}
\item The \emph{right profile of $A$ in $\set{F}$ over $\bool$} is the tuple

\centerline{$\hom_\bool(A, \set{F}) = (\hom_\bool(A, F_1), \etl, \hom_\bool(A, F_k))$.}
\end{itemize}

Let $A_1, \etl, A_n$ be instances whose active domains are pairwise disjoint.
\begin{itemize}
    \item 
The \emph{direct sum} $A_1 \dsum \etc \dsum A_n$ of $A_1, \etl, A_n$ is the instance such that
$R^{A_1 \dsum \etc \dsum A_n}=   R^{A_1} \union \etc \union R^{A_n}$,
 for every $R \in \sigma$. 
\item 
The \emph{direct product} $A_1 \dprod \etc \dprod A_n$ of $A_1, \etl, A_n$ is the instance such that the relation $R^{A_1 \dprod \etc \dprod A_n}$  consists of all tuples $(\mathbf{a}_1, \etl, \mathbf{a}_r)$ with $(\mathbf{a}_1(i), \etl, \mathbf{a}_r(i)) \in R^{A_i}$, for $1 \leq i \leq n$ and for every $R \in \sigma$ of arity $r$.
\end{itemize}

The next proposition is well known and has a straightforward proof.
\begin{proposition}\label{prop:hom-combinatorics}
Let $A,B_1,B_2$ be instances,
and let $K\in\{\bool,\nat\}$. Then the following statements are true.
\begin{enumerate}
\item
    $\hom_K(A,B_1\oplus B_2) = \hom_K(A,B_1)+_K\hom_K(A,B_2)$,
    provided that $A$ is connected;
\item
    $\hom_K(A,B_1\otimes B_2) = \hom_K(A,B_1)\cdot_K\hom_K(A,B_2)$;
\item
    $\hom_K(B_1\oplus B_2,A) = \hom_K(B_1,A)\cdot_K\hom_K(B_2,A)$,
\end{enumerate}
where  $+_\nat$ and $\cdot_\nat$ stand for  addition $+$  and multiplication $\times$ of non-negative integers, while
 $+_\bool$ and $\cdot_\bool$ stand for  disjunction $\lor$ 
and conjunction  $\land$ of  the Boolean values $0$ and $1$.
\end{proposition}

\section{Left Query Algorithms and Right Query Algorithms}\label{sec:left_and_right_alg_bool}

In \cite{chen2022algorithms}, Chen et al.\  focused on  classes of graphs and introduced the notions of a left query algorithm and a right query algorithm over the semiring $\nat$ of the non-negative integers.  
Here, we extend their framework in two ways: first, we consider classes of instances over some arbitrary, but fixed, schema; second, we consider left query algorithms and right query algorithms over the Boolean semiring $\bool$.

\begin{definition}\label{def:query_alg}
Let $\class{C}$ be a class of instances and let  $K$ be the semiring $\bool$ or $\nat$.
\begin{itemize}
\item Assume that $k$ is a positive integer. 
 \begin{itemize}
     \item 
    A \emph{left $k$-query algorithm over $K$ for $\class{C}$}  is a pair
$(\set{F},X)$, where $\set{F} = \sete{F_1, \etl, F_k}$ is a set of instances and $X$ is a set of $k$-tuples over $K$, such that
for all instances $D$, we have that $D \in \class{C}$ if and only if $\hom_K(\set{F}, D) \in X$. 
\item 
 A \emph{right $k$-query algorithm over $K$ for $\class{C}$}  is a pair
$(\set{F},X)$, where $\set{F} = \sete{F_1, \etl, F_k}$ is a set of instances and $X$ is a set of $k$-tuples over $K$, such that
for all instances $D$, we have that $D \in \class{C}$ if and only if 
$\hom_K(D, \set{F}) \in X$.     
\end{itemize}
\item We say that $\class{C}$ \emph{admits a left query algorithm over} $K$ if for some $k>0$, there is a left $k$-query algorithm over $K$ for $\class{C}$.
Similarly, we say that $\class{C}$ \emph{admits a right query algorithm over} $K$ if for some $k>0$, there is a right $k$-query algorithm over $K$ for $\class{C}$.
\end{itemize} 
\end{definition}

The term ``query algorithm'' is natural because
we can think of a query algorithm as a procedure for determining
if a given instance belongs to the class $\class{C}$:
we  compute the left homomorphism-count vector (in the case of a left query algorithm)  or the right homomorphism-count vector (in the case of a right query algorithm) 
and test whether it belongs to $X$. When the semiring $\nat$ is considered,
  this is a somewhat abstract notion of an  algorithm   because it makes no requirements on the effectiveness of the set $X$.
  Not requiring $X$ to be a decidable set makes our results regarding the non-existence of 
  left query algorithms over $\nat$ stronger. 
 Moreover, in all cases
where we establish the existence of a left query algorithm over $\nat$ or a right query algorithm over $\nat$,  the set $X$ will happen to be decidable (for the semiring $\bool$, the set $X$ is always finite, hence decidable).  

Let $K$ be the semiring $\bool$ or $\nat$. 
It is clear that if two classes of instances admit a left query algorithm over $K$, then so do their complements, their union, and their intersection. Consequently, the classes of instances that admit a left query algorithm over $K$ are closed under Boolean combinations. Furthermore, the same holds true for right query algorithms.

By Part 3 of Proposition 
\ref{prop:hom-combinatorics}, 
we have that if $K$ is the semiring $\bool$ or the semiring $\nat$, then  $\hom_K$ is \emph{multiplicative on the left}, i.e., for all instances $A_1, \etl, A_n$ and $B$, we have $\hom_K(A_1 \dsum \etc \dsum A_n, B) = \hom_K(A_1, B) \cdot_K\ \etc \cdot_K\ \hom_K(A_n, B)$. It follows that, as regards the existence of left query algorithms, we may assume that all instances in the finite set $\class{F}$ of a left query algorithm over $\bool$ or over $\nat$ are connected. We state this observation as a proposition that will be used repeatedly in the sequel.

\begin{proposition} \label{prop:connected}
Let $\class{C}$ be a class of instances and let $K$ be the semiring $\bool$ or $\nat$. Then the following statements are equivalent.
\begin{enumerate}
    \item $\class{C}$ admits a left query algorithm $(\set{F},X)$ over $K$.
    \item $\class{C}$ admits a left query algorithm $(\set{F},X)$ over $K$, where every instance in the set $\set{F}$ is connected.
\end{enumerate}
\end{proposition}

Left  profiles over $\nat$ contain more information than left  profiles over $\bool$.  Therefore, if membership in a class $\class{C}$ can be determined using left  profiles over $\bool$, then it ought to be also determined using left  profiles over $\nat$. 
Similar considerations hold for right profiles.
The next proposition makes these assertions precise.

\begin{proposition}\label{prop:query_alg_bool_implies_bag}
Let $\class{C}$ be a class of instances and let $\set{F}$ be a finite set of instances.

If   $\class{C}$ admits a left query algorithm over $\bool$ of the form $(\set{F},X)$ for some set $X$, then $\class{C}$  admits a left query algorithm over $\nat$ of the form $(\set{F},X')$
for some set $X'$. 
In particular, if $\class{C}$ admits a left query algorithm over $\bool$, then it also admits a left query algorithm over $\nat$.

Furthermore, the same holds true for right query algorithms.
\end{proposition}
\begin{proof}
 Assume that $\class{C}$ admits a left query algorithm $(\set{F}, X)$ over $\bool$, where $\set{F} = \sete{F_1, \etl, F_k}$ and $X \subseteq
 \{0,1\}^k$, for some $k>0$.
 For every  $t = (t_1, \etl, t_k) \in \{0,1\}^k$,  we let  $X_t$ be the set

\centerline{$X_t= \{(s_1, \etc, s_k)\in \nat^k: s_i=0 ~~ \mbox{if and only if} ~~ t_i=0,~\mbox{for}~1\leq i\leq k\}$.}

\noindent Consider the set $X' =  \bunion_{t \in X} X_t$. It is easy to verify that the pair $(\set{F}, X')$ is a left $k$-query algorithm for $\class{C}$ over $\nat$. The argument for right query algorithms is entirely analogous.
\end{proof}

As pointed out in the Introduction, the converse of Proposition \ref{prop:query_alg_bool_implies_bag} is not true, in general.

We now give several examples illustrating left and right query algorithms.

\begin{example}\label{exam:triangle-free}
Let $\mathcal C$ be the class of all triangle-free graphs, i.e., the graphs $G$ for which there is no homomorphism from $K_3$ to $G$. Clearly,  $\mathcal C$ admits a left $1$-query algorithm $(\set{F}, X)$ over $\bool$, where $\set{F}=\{K_3\}$ and
$X =\{0\}$ (recall that $K_3$ is the $3$-clique).
\footnote{This example, and several other examples in this paper, involve graphs.
Here, the word ``graph'' may be read as
``structure over a schema with one binary relation''. Alternatively, it may be read as
``structure over a schema with one binary relation that is symmetric and irreflexive'', but, in the latter case, we only require the query algorithm to behave correctly on such graphs, and we do not require the query algorithm to distinguish such graphs from structures whose relation is not symmetric and irreflexive. 
}
Therefore, by Proposition \ref{prop:query_alg_bool_implies_bag}, $\cal C$ admits a left $1$-query algorithm over $\nat$. In contrast, Chen et al.~\cite[Proposition~8.2]{chen2022algorithms} showed that (the complement of) $\mathcal C$ does not admit a right query algorithm over $\nat$, hence (again by Proposition \ref{prop:query_alg_bool_implies_bag}) it does not admit a right query algorithm over $\bool$.
\end{example}

We now recall the definition of
constraint satisfaction problems.

\begin{definition} \label{defn:csp}
If $B$ is an instance, then the \emph{constraint satisfaction problem} $\csp(B)$ is the following decision problem: given an instance $A$, is there a homomorphism from $A$ to $B$?
\end{definition}
For $k \geq 2$, let $K_k$ denote the $k$-clique.  Then $\csp(K_k)$ is the {\sc $k$-colorability} problem: given a graph $G$, is $G$ $k$-colorable? During the past three decades, there has been an extensive study of  complexity of constraint satisfaction problems, motivated by the Feder-Vardi Conjecture that for every instance $B$, either $\csp(B)$ is NP-complete or $\csp(B)$ is solvable in polynomial time. This conjecture was eventually confirmed independently by Bulatov~\cite{DBLP:conf/focs/Bulatov17} and Zhuk~\cite{DBLP:conf/focs/Zhuk17}.

Every constraint satisfaction problem will be identified with the class of its ``yes'' instances, that is, for every instance $B$, we have that $\csp(B) = \{A: A\rightarrow B\}$.

\begin{example} \label{exam:csp}
Let $B$ be an instance. Clearly,  $\csp(B)$ admits a right $1$-query algorithm $(\set{F}, X)$ over $\bool$, where $\set{F}=\{B\}$ and $X=\{1\}$. Therefore, by Proposition \ref{prop:query_alg_bool_implies_bag}, $\csp(B)$ admits a right $1$-query algorithm $(\set{F}, X')$ over $\nat$. In particular, the {\sc 3-colorability} problem $\csp(K_3)$ admits a right query algorithm over both $\bool$ and $\nat$. In contrast, it will follow from results in Section~\ref{sec:left_alg_bool} and Section~\ref{sec:bool_vs_bag}  that
$\csp(K_3)$ does not admit a left query algorithm over $\bool$ or over $\nat$ (see~Remark~\ref{rem:no-lqa}).
\end{example}

\begin{example} \label{exam:hom-equiv-K3}
Consider the homomorphic equivalence class $\homtype{K_3}$. Note that $\homtype{K_3}$ is the class of all graphs that are $3$-colorable and also contain a triangle.  From results in Section \ref{sec:left_alg_bool}, it will follow that $\homtype{K_3}$ does not admit a left query algorithm over $\bool$ (see~Remark~\ref{rem:no-lqa}). Furthermore, from results in Section \ref{sec:right_alg}, it will follow that $\homtype{K_3}$ does not admit a right query algorithm over $\bool$ 
(see~Remark~\ref{rem:no-rqa}).
\end{example}

\section{Left Query Algorithms over \texorpdfstring{$\bool$}{B}}\label{sec:left_alg_bool}
In this section, we investigate which classes admit a left query algorithm over $\bool$. It is easy to see that every class of instances that admits a left query algorithm over $\bool$ is closed under homomorphic equivalence. In other words, closure under homomorphic equivalence is a necessary condition for the existence of a
left query algorithm over $\bool$. The next result gives an exact characterization of the classes of instances that admit a left query algorithm over $\bool$.

\begin{restatable}{theorem}{thmcharleftalgbool}
\label{thm:char-left-alg-bool}
Let $\class{C}$ be a class of instances. Then the following statements are equivalent.
\begin{enumerate}
\item $\class{C}$ admits a left query algorithm over $\bool$.
\item $\class{C}$ is definable by a Boolean combination of CQs.
\item $\class{C}$ is FO-definable and closed under homomorphic equivalence.
\end{enumerate}
\end{restatable}
\begin{proof}
We will first show the equivalence between statements (1) and (2), and then the equivalence between statements (2) and (3).

\noindent $(1) \Longrightarrow (2)$: Let a left query algorithm over $\bool$ for $\class{C}$ consist of $\set{F} = \sete{F_1, \etl, F_k}$ and $X \subseteq \cartpwr{\sete{0, 1}}{k}$, and let $\canq{F_1}, \etl, \canq{F_k}$ be the canonical conjunctive queries of $F_1, \etl, F_k$, respectively. For every tuple $t = (t_1, \etl, t_k) \in X$, define
\[\varphi_t \defas \bland_{t_i = 1} \canq{F_i} \land \bland_{t_i = 0} \neg\canq{F_i}.\]
Take $\varphi \defas \blor_{t \in X} \varphi_t$.
Then $\class{C} = \modclass(\varphi)$. 

\noindent $(2) \Longrightarrow (1)$: Every class $\class{C}$ defined by a conjunctive query $q$ admits a left $1$-query algorithm over $\bool$. Indeed, we can pick $\set{F}$ to consist of the canonical instance $A^q$ of $q$, and $X=\{1\}$. It follows by closure under Boolean combinations that every class defined by a Boolean combination of conjunctive queries also admits a left query algorithm over $\bool$.

\noindent $(2) \Longrightarrow (3)$: This implication  is immediate because CQs are first-order formulas whose truth is  preserved by homomorphisms.

\noindent $(3) \Longrightarrow (2)$:
We will use two results from \cite{rossman2008homomorphism} about the homomorphism preservation theorem in the finite. For instances $A$ and $B$, we write $A \leftrightarrow^n B$ to mean that $A$ and $B$ satisfy the same existential positive FO-sentences of quantifier rank at most $n$, and write $A \equiv^n B$ to mean that $A$ and $B$ satisfy the same FO-sentences of quantifier rank at most $n$. The two results from \cite{rossman2008homomorphism} in our concerns are
\begin{enumerate}[(a)]
\item Theorem~1.9: For every $n$, there is some $m$ that depends on $n$ such that for every instances $A$ and $B$ with $A \leftrightarrow^m B$, there are instances $A'$ and $B'$ such that $A \homeq A'$, $B \homeq B'$, and $A' \equiv^n B'$.
\item Lemma~3.9: For every $m$, the equivalence relation $A \leftrightarrow^m B$ has finitely many equivalence classes over the class of all instances.
\end{enumerate}

Now, let $\class{C}$ be a class definable by a FO-sentence $\varphi$ and closed under homomorphic equivalence $\homeq$. Let $n$ be the quantifier rank of $\varphi$, and let $m$ be the integer in the statement of (a). Note that $m$ depends on $n$ only.

We claim that $\modclass(\varphi) = \class{C}$ is closed under $\leftrightarrow^m$. Indeed,
assume that $A$ and $B$ are two instances such that $A \satis \varphi$ and $A \leftrightarrow^m B$. By (a), there are instances $A'$ and $B'$ such that $A \homeq A'$, $B \homeq B'$, and $A' \equiv^n B'$. It follows, successively, that\\
\centerline{
\begin{tabular}{rl}
$A' \satis \varphi$ & (since $A \satis \varphi$, $A \homeq A'$, and $\class{C}$ is closed under $\homeq$), \cr
$B' \satis \varphi$ & (since $\varphi$ has quantifier rank $n$ and $A' \equiv^n B'$), \cr
$B \satis \varphi$  & (since $B \homeq B'$ and $\class{C}$ is closed under $\homeq$).
\end{tabular}}

By (b), the equivalence relation $\leftrightarrow^m$ has finitely many equivalence classes. Let $A_1, \etl, A_k$ be representatives from each of these equivalence classes (one per equivalence class). For different $i, j$ in $\sete{1, \etl, k}$, let $\psi_{i, j}$ be an existential positive FO-sentence of quantifier rank at most $m$ such that $A_i \satis \psi_{i, j}$ but not $A_j \satis \psi_{i, j}$, and let $\psi(A_i)$ be the conjunction of all $\psi_{i, j}$. Each $\psi(A_i)$ is a Boolean combination of conjunctive queries and it holds for every instance $B$ that $B \leftrightarrow^m A_i$ if and only if $B \satis \psi(A_i)$. Then $\varphi$ is equivalent to the disjunction $\blor_{A_i \satis \varphi} \psi(A_i)$ since $\modclass(\varphi)$ is closed under $\leftrightarrow^m$. Indeed, if $B \satis \varphi$, then $B \leftrightarrow^m A_i$ for some $A_i$ that satisfies $\varphi$, hence $B \satis \psi(A_i)$. Conversely, if $B \satis \blor_{A_i \satis \varphi} \psi(A_i)$, then $B \satis \psi(A_i)$ for some $A_i$ that satisfies $\varphi$; it follows that $B \leftrightarrow^m A_i$, hence $B \satis \varphi$.
\end{proof}

\begin{corollary}\label{cor:lqa-for-fo-classes}
A class $\class{C}$ of instances  that is closed under homomorphic equivalence admits a 
left query algorithm over $\bool$ if and only if $\class{C}$ is FO-definable.
\end{corollary}

Corollary~\ref{cor:lqa-for-fo-classes}, in particular,
applies to classes of the form $\csp(A)$, since
such classes are closed under homomorphic equivalence. 
It was shown in \cite{tardif2007characterisation} that testing, for a 
given instance $A$, whether $\csp(A)$ is
FO-definable, is NP-complete (and in fact, in polynomial time when $A$ is a \emph{core}, i.e. when there is no homomorphism from $A$ to a proper subinstance of $A$). It follows that testing if a 
    given $\csp(A)$ admits a left query algorithm over $\bool$
    is NP-complete. This extends to finite unions
    of CSPs:

\begin{proposition}\label{prop:u-csp-left}
    The following problem is NP-complete: 
    given instances $A_1, \ldots, A_n$, does $\bigcup_{1\leq i \leq n} \csp(A_i)$
    admit a left query algorithm over $\bool$? 
\end{proposition}
\begin{proof}
    Without loss of generality, assume that the instances 
    $A_1, \ldots, A_n$ are pairwise 
    homomorphically incomparable  (because, if
    $A_j\to A_k$, then $A_j$ can be removed without
    affecting the class defined by $\bigcup_{1\leq i\leq n}\csp(A_i)$). It is known that, in this case,
    $\bigcup_{1\leq i\leq n}\csp(A_i)$ is FO-definable if and only if
    for each $1 \leq i\leq n$, $\csp(A_i)$ is FO-definable
    (this may be considered folklore, see~\cite[Lemma 5.13]{Bienvenu15:ontology} for an explicit proof). 
    The result follows, by Corollary~\ref{cor:lqa-for-fo-classes}.
\end{proof}

Corollary~\ref{cor:lqa-for-fo-classes} also applies
to classes of the form $\homtype{A}$, and
we can derive a
similar complexity bound. This will be obtained using  the 
following proposition.

\begin{proposition}\label{prop:positive-homtype-csp}
    Let $A$ be an instance and 
    $\mathit{K} \in \sete{\bool, \nat}$. 
    Then the following statements
    are equivalent.
    \begin{enumerate}
        \item $\homtype{A}$ has a left query algorithm over $\mathit{K}$.
        \item $\csp(A)$ has a left query algorithm over $\mathit{K}$.
    \end{enumerate}
\end{proposition}

\begin{proof}
Let $\mathit{K} \in \sete{\bool, \nat}$ throughout the proof.

\noindent  $(2) \Longrightarrow (1)$:
Clearly, 
$\homtype{A} = \csp(A) \intsc \setm{B}{A \homto B}$. 
Also, $\setm{B}{A \homto B}$
 admits an obvious left query algorithm over $\mathit{K}$. The result follows
by closure under Boolean combinations.

\noindent  $(1) \Longrightarrow (2)$: Assume $\csp(A)$ has no  left query algorithm over $\mathit{K}$.
Let $\set{F}$ be an arbitrary finite set of connected instances (think: candidate left query algorithm for $\homtype{A}$).
Since $\csp(A)$ has no left query algorithm over $\mathit{K}$, there are instances $P \in \csp(A)$ and $Q \notin \csp(A)$ 
such that $\hom_\mathit{K}(\set{F},P) = \hom_\mathit{K}(\set{F},Q)$.
Let $P' = P \dsum A$ and let $Q' = Q \dsum A$.
Then, by construction, $P' \in \homtype{A}$ and $Q' \notin \homtype{A}$ but $\hom_\mathit{K}(\set{F},P') = \hom_\mathit{K}(\set{F},Q')$ (cf.~Proposition~\ref{prop:hom-combinatorics}).
Therefore, by Proposition~\ref{prop:connected}, $\homtype{A}$ has no left query algorithm over $\mathit{K}$.
\end{proof}

Consequently, the following problem is also NP-complete: 
\emph{given an instance $A$, does $\homtype{A}$
    admit a left query algorithm over $\bool$?} 
    
    \begin{remark}\label{rem:no-lqa}
In Example~\ref{exam:csp}, we asserted that 
the class $\csp(K_3)$ of 3-colorable graphs admits no  left query algorithm over $\bool$.
Furthermore, in Example~\ref{exam:hom-equiv-K3},
we asserted that $\homtype{K_3}$ (i.e.,
the class of graphs that are 3-colorable and also contain a triangle)  admits  no left query algorithm over $\bool$.
Corollary~\ref{cor:lqa-for-fo-classes} and Proposition~\ref{prop:positive-homtype-csp},
now account for the non-existence of a left query algorithm over
$\bool$ for these classes: the reason is these two classes are not FO-definable. In the next section
(see~Corollary~\ref{cor:lqa-nat-fo}), we will see that
the same explanation accounts for the fact that these
classes admit no left query algorithm over $\nat$ either.
\end{remark}

    Proposition \ref{prop:positive-homtype-csp}
    extends  to finite unions of homomorphic equivalence classes.

\begin{theorem}\label{thm:hom-equiv-class-left}
For all instances $A_1, \etl, A_n$, the following statements are equivalent.
\begin{enumerate}
\item $\bunion_{1 \leq i \leq n} \homtype{A_i}$ admits a left query algorithm over $\bool$ (equivalently, is FO-definable).
\item Each $\homtype{A_i}$, for $i=1,\ldots, n$, admits a left query algorithm over $\bool$ (equivalently, is FO-definable).
\end{enumerate}
\end{theorem}

\begin{proof} {} 
It is clear that the second statement implies the first.
We will prove by induction on $n$ that the first statement implies the second.
The base case ($n = 1$) is immediate since (i) and (ii) coincide. Next, let $n > 1$ and $\class{C} \defas \bunion_{1 \leq i \leq n} \homtype{A_i}$. 
We proceed by contraposition: suppose that $\homtype{A_i}$ does not admit a left query algorithm over $\bool$, for some $i \leq n$.
We may assume without loss of generality that 
$A_1, \ldots, A_n$ are pairwise not homomorphically equivalent.
Note that $\homto$ induces a preorder among $A_1, \etl, A_n$ and, since $n$ is finite, there is a maximal. Without loss of generality, assume that $A_n$ is a maximal, that is, $A_n \not\homto A_i$ for all $i < n$. We distinguish two cases.

(1) $\homtype{A_n}$ admits a left query algorithm over $\bool$. Then, for some $i\leq n - 1$, $\homtype{A_i}$ 
does not admit a left query algorithm over $\bool$.
By induction hypothesis, we have $\class{C}' \defas \bunion_{1 \leq i \leq n - 1} \homtype{A_i}$ does not admit a left query algorithm over $\bool$. Then it follows that $\class{C}$ does not admit a left query algorithm over $\bool$ either, for otherwise $\class{C}' = \class{C} \setminus \homtype{A_n}$ would admit a 
left query algorithm over $\bool$. 

(2) $\homtype{A_n}$ does not admit a left query algorithm over $\bool$. By Proposition~\ref{prop:positive-homtype-csp}, $\csp(A_n)$ does not admit a left query algorithm over $\bool$ either. Let $\set{F}$ be an arbitrary finite non-empty set of connected instances. Since $\csp(A_n)$ does not admit a left query algorithm over $\bool$, there are $P \in \csp(A_n)$ and $Q \notin \csp(A_n)$ such that $\hom_\bool(\set{F}, P) = \hom_\bool(\set{F}, Q)$. It follows that
\begin{itemize}
\item $(P \dsum A_n) \in \homtype{A_n}$, because both $P, A_n \in \csp(A_n)$,
\item $(Q \dsum A_n) \notin \homtype{A_n}$, because $Q \notin \csp(A_n)$,
\item for all $i < n$, $(Q \dsum A_n) \notin \homtype{A_i}$, because $A_n \not\homto A_i$, and
\item $\hom_\bool(\set{F}, P \dsum A_n) = \hom_\bool(\set{F}, Q \dsum A_n)$, because the instances in $\set{F}$ are all connected.
\end{itemize}
The first three points above give $(P \dsum A_n) \in \class{C}$ and $(Q \dsum A_n) \notin \class{C}$. Therefore, by Proposition~\ref{prop:connected}, $\class{C}$ has no left query algorithm over $\bool$. 
\end{proof}

Note that Theorem~\ref{thm:hom-equiv-class-left} only applies to finite unions of
homomorphic-equivalence classes. It may fail for 
infinite unions. Specifically, the class of all instances trivially admits a left query algorithm over $\bool$,
and it is the union of all equivalence classes $\homtype{A}$, as  $A$ varies over all instances; however, as seen earlier, $\homtype{K_3}$ does not admit any left query algorithm over $\bool$.

\begin{corollary}
    The following problem is NP-complete: 
    given instance $A_1, \ldots, A_n$, does $\bigcup_{1 \leq i \leq n} \homtype{A_i}$
    admit a left query algorithm over $\bool$?
\end{corollary}

Each class $\class{C}$ that is closed under
homomorphic equivalence  can trivially be 
represented as a possibly-infinite union of classes
of the form $\homtype{A}$. The 
algorithmic problem of testing for the existence
of a left query algorithm, of course, makes sense only for 
finitely presented inputs. This motivates
the above corollary.

As a last case study, 
consider Boolean 
Datalog programs, i.e.,
Datalog programs with a zero-ary goal predicate.
We  omit a detailed definition of the syntax and semantics of Datalog, which can be found, e.g., in \cite{abiteboul1995foundations}. 
Each Datalog program $P$ naturally defines a class of instances
$\class{C}_P$. 
It is well known that \emph{the 
class
$\class{C}_P$ is closed under homomorphic equivalence}.
Furthermore, $\class{C}_P$ is 
FO-definable if and only if $P$ is ``bounded'' (meaning that 
there is a fixed number $n$, depending only on $P$
 and not on the input instance, such that
$P$ reaches its fixed point after at most 
$n$ iterations), as was first shown by 
\cite{AjtaiGurevich94} and also follows from \cite{rossman2008homomorphism}.
The boundedness problem for Boolean Datalog programs is undecidable~\cite{Gaifman93:undecidable}. 
Therefore, we have the following result.

\begin{corollary}
    The following problem is undecidable:
    given a Boolean Datalog program $P$, 
    does $\class{C}_P$ admit a left query algorithm over $\bool$.
\end{corollary}

\section{Existence vs.~Counting: When Does Counting Not Help?}\label{sec:bool_vs_bag}

Left query algorithms over $\nat$ are more powerful than left query algorithms
over $\bool$. This is trivially so
because query algorithms over $\bool$
cannot distinguish homomorphically equivalent
instances. 

\begin{example}
Let $\class{C}$ be the isomorphism class of the 
instance $A$ consisting of the fact $R(a,a)$
(in other words, the
single-node reflexive digraph).
Since
$\class{C}$ is not closed under homomorphism
equivalence, it does not admit a left query algorithm (nor a right query algorithm) over $\bool$. 
On the other hand, it admits a left query algorithm
over $\nat$: by counting the number of
homomorphisms from $A$ to a given input instance $B$, we can verify that $B$ contains a reflexive
node; and by counting the number of homomorphisms
from the instance $A'$ consisting of a single edge $R(a,b)$, we can verify that the total number of
edges in the graph is equal to 1. More precisely,
let $\class{F}=\{A,A'\}$ and let $X=\{(1,1)\}$.
Then $(\class{F},X)$ is a left query algorithm over $\nat$
for $\class{C}$.
\end{example}

As it turns out, in a precise sense, this is the only reason why 
left query algorithms over $\nat$ are more powerful than
left query algorithms over $\bool$: the ability to count
does not give more power when it comes to classes that are closed under homomorphic equivalence. This follows from 
the next theorem.

\begin{theorem}\label{thm:lqa-nat-fo1}
Let $\class{C}$ be a class of instances closed under
homomorphic equivalence. For every finite set $\set{F}$ of connected instances, the following statements are equivalent.
\begin{enumerate}
    \item $\class{C}$ admits a left query algorithm over $\nat$ of the form $(\set{F},X)$ for some set $X$.
    \item $\class{C}$ admits a left query algorithm over $\bool$ of the form $(\set{F},X')$ for some set $X'$.
\end{enumerate}
\end{theorem}

\begin{proof}
The case $\class{C} = \emptyset$ is trivial, so we will assume that $\class{C}$ is non-empty.
The implication $(2) \Longrightarrow (1)$ is given by  Proposition~\ref{prop:query_alg_bool_implies_bag}.
Let us prove the implication $(1)  \Longrightarrow (2)$.

Let $(\class{F},X)$ be a left query algorithm over $\nat$ for $\class{C}$  where $\class{F}=\{F_1,\dots,F_k\}$
of pairwise non-isomorphic instances and $X\subseteq\nat^k$. It is enough to focus on {\em simple} sets $X$, where a set $X$ is simple if 
for every $1\leq i\leq k$ and every $\mathbf{t},\mathbf{t}'\in X$, $\mathbf{t}(i)=0$ iff $\mathbf{t}'(i)=0$.  Indeed, assume that $(1)\Rightarrow(2)$ holds whenever $X$ is simple. Then, if $X$ is not simple, partition $X$ into maximal simple subsets $X_1,\dots,X_r$ and, for every $1\leq i\leq r$, let $\class{C}_i$ be the class of instances that admits the left query algorithm $(\class{F},X_i)$. It is easy to verify that $\class{C}_i$ is closed under homomorphic equivalence and, hence, by assumption, admits a left query algorithm over $\bool$ of the form $(\set{F}, X_i')$ for some set $X_i'$. Then $(\set{F}, \bunion_{1 \leq i \leq r} X_i')$ is a left query algorithm over $\bool$ for $\class{C}$.

Let us assume that $X$ is simple and non-empty (otherwise $\class{C}=\emptyset$, contradicting our assumption) and let $\mathbf{t}\in X$. We can assume, by reordering the instances in $\class{F}$ if necessary, that there exists $s\geq 0$, such that $\mathbf{t}(i)>0$ for every $i\leq s$ and $\mathbf{t}(i)=0$ for every $i>s$.

Consider the FO-sentence defined as:
 
$$\varphi=\bigwedge_{i\leq s} q^{F_i}
\land\bigwedge_{i>s} \neg q^{F_i}$$

We shall prove that $\modclass(\varphi)=\class{C}$,
and therefore $\class{C}$ admits a left query algorithm over $\bool$ (namely, the left query algorithm $(\set{F},\{\mathbf{t}'\})$ where
$\mathbf{t}'(i)=1$ for $i\leq s$ and $\mathbf{t}'(i)=0$ for $i>s$). Since $\modclass(\varphi)=\class{C}$ already holds when $s = 0$, we assume $s > 0$ in the sequel.

To this end we shall need the following two technical lemmas.

\begin{restatable}{lemma}{lepol}
\label{le:pol}
Let $\mathbf{t}_1,\dots,\mathbf{t}_r\in\nat^s$ satisfying the following conditions:
\begin{enumerate}
\item For every  $1\leq i\leq s$, there exists $1\leq j\leq r$ such that $\mathbf{t}_j(i)\neq 0$.
\item For every different $i,i'\in\{1,\dots,s\}$, there exists $1\leq j\leq r$ such that $\mathbf{t}_j(i)\neq \mathbf{t}_j(i')$.
\end{enumerate}
Then there exists some integer $c>0$ such that for every $\mathbf{b}\in\mathbb{Z}^s$ whose entries are integers divisible by $c$, there is a multivariate polynomial with integer coefficients $p(x_1,\dots,x_r)$  of degree $s$ satisfying $p(0,\dots,0)=0$ (that is, such that the independent term is zero) and such that, for each $1\leq i\leq s$, $p(\mathbf{t}_1(i),\dots,\mathbf{t}_r(i))=\mathbf{b}(i)$.
\end{restatable}
\begin{proof}
Most likely this  follows from known results but, in any case, we provide a self-contained proof. Let $\class{G}$ be the set of all linear combinations $a_1\mathbf{t}_1+\cdots a_r\mathbf{t}_r$ where each $a_i$ is a non-negative integer. It follows from condition (1) that $\class{G}$ contains a tuple where all entries are positive. Moreover it follows from (2) that $\class{G}$ contains a tuple where all entries are different and positive. Indeed, if $\mathbf{t},\mathbf{u}\in\class{G}$ and $\mathbf{v}=d\mathbf{t}+\mathbf{u}$ for $d\in\nat$ large enough, then for every different $i,i'\in\{1,\dots,s\}$, we have\\
\centerline{($\mathbf{t}(i)\neq \mathbf{t}(i')$ \ or \ $\mathbf{u}(i)\neq \mathbf{u}(i')$) \ implies \ $\mathbf{v}(i)\neq  \mathbf{v}(i')$.}
Now, let $\mathbf{u} = a_1\mathbf{t}_1+\cdots+a_r\mathbf{t}_r$ be any tuple where all entries are different and positive. It is easy to see that the $(s\times s)$-matrix $M$ whose rows are  $\mathbf{u}^1,\mathbf{u}^2,\dots,\mathbf{u}^s$ is non-singular. Indeed, if  $\mathbf{u}=(u_1,\dots,u_s)$, then $\text{det}(M)=u_1\cdots u_s\cdot \text{det}(N)$, where $N$ is the $(s\times s)$-matrix with rows 
$\mathbf{u}^0,\mathbf{u}^1,\dots,\mathbf{u}^{s-1}$ which is Vandermonde. It is well known that Vandermonde matrices are non-singular whenever $\mathbf{u}$ has no repeated elements (see \cite{Kalman84} for example). 

To finish the proof we choose $c$ to be $|\text{det}(M)|$.  By assumption all entries of $\mathbf{b}$ are divisible by $c$ which implies that all entries of $M^{-1}\textbf{b}$ are integers, that is, $\mathbf{b}$ can be expressed as $e_1 \mathbf{u}^1+\cdots+e_s\mathbf{u}^s$ for some $e_1,\dots,e_s\in\mathbb{Z}$.
Hence, the polynomial $p(x_1,\dots,x_r)=e_1 y^1+\cdots+e_s y^s$ where $y=a_1x_1+\cdots + a_rx_r$ satisfies the claim.
\end{proof}

\begin{restatable}{lemma}{lelovasz}
\label{le:lovasz}
Let $F,F'$ be instances such that there is no surjective homomorphism from $F$ onto $F'$. Then there exists some subinstance $H$ of $F'$ such that 
$\hom_{\nat}(F,H) \neq \hom_{\nat}(F',H)$.
\end{restatable}

\begin{proof}
This is a natural adaptation of the Lov\'asz's proof that two instances are isomorphic if and only if they have the same left homomorphism-count vector.

For every instance $G$, we write $\sur_\nat(G, F')$ for the number of surjective homomorphisms from $G$ onto $F'$; moreover, for every subset $S \subseteq \adom(F')$, we write $\hom_\nat^S(G, F')$ for the number of homomorphisms $h: G \homto F'$  whose range is contained in $S$. Let $n \defas \size{\adom(F')}$. By the Inclusion-Exclusion Principle, then
\[
\sur_\nat(G, F') = \sum_{S \subseteq \adom(F')} (-1)^{n - \size{S}} \hom_\nat^S(G, F').
\]
Since $\sur_\nat(F, F') = 0$ and $\sur_\nat(F', F') > 0$, we have $\sur_\nat(F, F') \neq \sur_\nat(F', F')$ and it follows by the above discussion that there is a subset $S \subseteq \adom(F')$ such that $\hom_\nat^S(F, F') \neq \hom_\nat^S(F', F')$. Let $H$ be the maximum subinstance of $F'$ with $\adom(H) \subseteq S$, then we have $\hom_\nat(F, H) = \hom_\nat^S(F, F') \neq \hom_\nat^S(F', F') = \hom_\nat(F', H)$.
\end{proof}

Let us now continue with the proof that 
$\modclass(\varphi)=\class{C}$.
The direction $\supseteq$ is immediate. For the converse, we must prove that every $B\in\modclass(\varphi)$ belongs also to $\class{C}$. 
Let $B\in\modclass(\varphi)$, and 
choose an arbitrary $A \in \class{C}$ (by the assumption that $\class{C} \neq \emptyset$).
Let $\mathbf{t}^A = \hom_\nat(\set{F}, A)$ and $\mathbf{t}^B = \hom_\nat(\set{F}, B)$. Note that $\mathbf{t}^A(i)=\mathbf{t}^B(i)=0$ for every $i>s$ and $\mathbf{t}^A(i)$ and $\mathbf{t}^B(i)$ are strictly positive for every $i\leq s$.

Let $\class{H}=\{H_1,\dots,H_r\}$ be the non-empty set of instances constructed as follows:
\begin{enumerate}
    \item[(i)] For every $1\leq i\leq s$, $F_i$ is contained in $\class{H}$.
    \item[(ii)] For every $i,i'\in\{1,\dots,s\}$ such that there is no surjective homomorphism $F_i\rightarrow F_{i'}$, $\class{H}$ contains the instance $H$ given by Lemma \ref{le:lovasz}.
\end{enumerate}

For each $1\leq i\leq r$, let $\mathbf{t}_i = \hom_\nat(\sete{F_1, \etl, F_s}, H_i)$.
It can be readily verified that $\mathbf{t}_1,\dots,\mathbf{t}_r$ satisfy the conditions of Lemma \ref{le:pol}.
Condition (1) is guaranteed due to step (i) in the construction of $\class{H}$.
For condition (2), let $i,i'$ be any pair of different integers in $\{1,\dots,s\}$. Since $F_i$ and $F_{i'}$ are not isomorphic, it must be the case that there is no surjective homomorphism $F_i\rightarrow F_{i'}$ or there is no surjective homomorphism $F_{i'}\rightarrow F_i$. Hence, due to step (ii), $\class{H}$ contains some instance $H_j$ witnessing $\mathbf{t}_j(i)\neq \mathbf{t}_j(i')$. 

Let $c>0$ be given by Lemma \ref{le:pol}, let $\mathbf{b}\in\mathbb{Z}^s$ where $\mathbf{b}(i)=c(\mathbf{t}^B(i)-\mathbf{t}^A(i))$ for every $1\leq i\leq s$, and let $p(x_1,\dots,x_r)$ be the polynomial given by Lemma \ref{le:pol} for $\mathbf{b}$. We can express $p(x_1,\dots,x_r)$ as $p_A(x_1,\dots,x_r)-p_B(x_1,\dots,x_r)$ where all the coefficients in $p_A$ and $p_B$ are positive.

For every polynomial $q(x_1,\dots,x_r)$ where the independent term is zero, consider the instance $H_q$ defined inductively as follows:
\begin{itemize}
    \item If $q=x_j$, then $H_q$ is $H_j$.
\item If $q=u+v$, then $H_q$ is $H_u\oplus H_v$.
\item If $q=u\cdot v$, then $H_q=H_u\otimes H_v$.
\end{itemize}

\begin{lemma}
\label{le:Hq}
Let $H_q$ be constructed as above. Then:
\begin{enumerate}
    \item $H_q\rightarrow F_1\oplus\cdots\oplus F_s$.
    \item Let $F$ be a connected instance and let $\mathbf{t} = \hom_\nat(F, \class{H})$. Then $\hom_{\nat}(F,H_q)=q(\mathbf{t})$.
\end{enumerate}
\end{lemma}

This lemma directly follows from the definition of $H_q$. More precisely, the first item follows from the fact that $H_j\rightarrow F_1\oplus\cdots\oplus F_s$ for every $H_j\in\class{H}$ (as can be shown by a straightforward
induction argument). The second item follows from the definition of $H_q$ and Proposition \ref{prop:hom-combinatorics}.

Let $A'=A_1\oplus\cdots\oplus A_c\oplus H_{p_A}$ where $A_1,\dots,A_c$ are disjoint copies of $A$. Similarly, define $B'=B_1\oplus\cdots\oplus B_c\oplus H_{p_B}$ where
$B_1,\dots,B_c$ are disjoint copies of $B$. We claim that $\mathbf{t}^{A'} = \hom_\nat(\set{F}, A')$ and $\mathbf{t}^{B'} = \hom_\nat(\set{F}, B')$ coincide.

Let $1\leq i\leq k$ and consider first the case $i>s$. It follows from Lemma \ref{le:Hq}(1) that $F_i\not\rightarrow H_{p_A}$. Indeed, if $F_i\rightarrow H_{p_A}$, then $F_i\rightarrow F_{i'}$ for some $i'\leq s$ implying that $\class{C}=\emptyset$, contradicting our assumption. Similarly $F_i\not\rightarrow H_{p_B}$.
Consequently,
$\mathbf{t}^{A'}(i)=c\cdot \mathbf{t}^A(i)=0=c\cdot \mathbf{t}^B(i)=\mathbf{t}^{B'}(i)$. For the other case, namely $i \leq s$, we have
    \begin{align*}
    \mathbf{t}^{B'}(i)-\mathbf{t}^{A'}(i) &=\hom_{\nat}(F_i,B')-\hom_{\nat}(F_i,A') \\
    &=c\cdot \hom_{\nat}(F_i,B)-c\cdot \hom_{\nat}(F_i,A'))+\hom_{\nat}(F_i,H_{p_B})-\hom_{\nat}(F_i,H_{p_A}) \\
    &=c\cdot (\mathbf{t}^B(i)-\mathbf{t}^A(i))+\hom_{\nat}(F_i,H_{p_B})-\hom_{\nat}(F_i,H_{p_A}) \\
    &=c\cdot (\mathbf{t}^B(i)-\mathbf{t}^A(i))+p_B(\mathbf{t}_1(i),\dots,\mathbf{t}_r(i))-p_A(\mathbf{t}_1(i),\dots,\mathbf{t}_r(i)) \\
    &=c\cdot (\mathbf{t}^B(i)-\mathbf{t}^A(i))-p(\mathbf{t}_1(i),\dots,\mathbf{t}_r(i)) \\
    &=c\cdot (\mathbf{t}^B(i)-\mathbf{t}^A(i))-\mathbf{b}(i)=0,
    \end{align*}
    where $\hom_{\nat}(F_i,H_{p_A})=p_A(\mathbf{t}_1(i),\dots,\mathbf{t}_r(i))$ and $\hom_{\nat}(F_i,H_{p_B})=p_B(\mathbf{t}_1(i),\dots,\mathbf{t}_r(i))$ hold
    by Lemma \ref{le:Hq}(2).

To finish the proof, note that since $\mathbf{t}^A(i)>0$ for every $1\leq i\leq s$ it follows by Lemma \ref{le:Hq}(1) that $A'\rightarrow A$, and, hence $A\leftrightarrow A'$. Similarly  we have $B'\leftrightarrow B$. Since $\class{C}$ is closed under homomorphic equivalence we have that $A'\in\class{C}$. Since $\mathbf{t}^{A'}=\mathbf{t}^{B'}$ it follows that $B'$ and, hence, $B$ belong to $\class{C}$ as well.
\end{proof}    

\begin{corollary}\label{cor:lqa-nat-fo}
Let $\class{C}$ be a class of instances closed under
homomorphic equivalence. Then the following statements are equivalent.
\begin{enumerate}
    \item $\class{C}$ admits a left query algorithm over $\nat$.
    \item $\class{C}$ admits a left query algorithm over $\bool$.
    \item $\class{C}$ is FO-definable.
\end{enumerate}
\end{corollary}

The equivalence of statements (1) and (2) follows from Theorem~\ref{thm:lqa-nat-fo1} together with Proposition~\ref{prop:connected}.
The equivalence of statements (2) and (3) was already established in 
Corollary~\ref{cor:lqa-for-fo-classes}.

We would like to point out some interesting
special cases of Theorem~\ref{thm:lqa-nat-fo1}.
The first pertains to CSPs. Let us say that
a constraint satisfaction problem $\csp(B)$
is ``\emph{determined by $\hom_K(\set{F},\cdot)$}''
for some $K\in\{\bool,\nat\}$ and some class $\set{F}$ of
instances, if the following holds
for all instances $A$ and $A'$: if $\hom_K(\set{F},A)=\hom_K(\set{F},A')$
then $A\in\csp(B)$ iff $A'\in\csp(B)$. 
Then Theorem~\ref{thm:lqa-nat-fo1} can be rephrased as follows: 
for every finite set of connected instances $\set{F}$,
every CSP determined by $\hom_\nat(\set{F},\cdot)$ is
determined by $\hom_\bool(\set{F},\cdot)$. In contrast,
for $\set{T}$  the infinite class of all trees,  
the CSPs determined by $\hom_\bool(\set{T},\cdot)$ 
form a proper subclass of those determined by $\hom_\nat(\set{T},\cdot)$.%
\footnote{
Note that the definitions of $\hom_\nat(\set{F},\cdot)$ and $\hom_\bool(\set{F},\cdot)$ extend naturally to infinite classes $\class{F}$.}
This follows from results in \cite{Butti2021:fractional},
because the former are precisely the CSPs that can be solved using
arc-consistency, while the latter are precisely the CSPs that 
can be solved using basic linear programming relaxation (BLP).
See~\cite[Example 99]{Kun2016:new} for an example of a CSP that
can be solved using BLP but not using arc-consistency.

The second special case
pertains  to homomorphic-equivalence classes. Given the 
importance of homomorphic equivalence as a notion of 
equivalence in database theory, it is natural to ask
when a database instance $A$ can be uniquely identified 
up to homomorphic equivalence by means of a 
left query algorithm. As we saw earlier, in Section~\ref{sec:left_alg_bool}, for left query algorithm over
$\bool$, this is the case if and only if $\csp(A)$ is
FO-definable (a condition that can be tested effectively, and,
in fact, is NP-complete to test).
It follows from Corollary~\ref{cor:lqa-nat-fo}
that the same criterion determines whether $A$ can be uniquely
identified up to homomorphic equivalence by means of a
left query algorithm over $\nat$. This extends naturally to 
finite unions of homomorphic equivalence classes.

Finally, let us consider again classes $\class{C}$ 
defined by a Boolean Datalog program $P$. 
It follows from the results we mentioned in 
Section~\ref{sec:left_alg_bool} that 
such a class of instances $\class{C}$ admits a left query algorithm over $\nat$
if and only if $P$ is bounded, and that testing
for the existence of a left query algorithm over $\nat$ is 
undecidable.

\section{Right Query Algorithms}\label{sec:right_alg}

Just as for left query algorithms, 
we have that 
every class of instances that admits a right query algorithm over $\bool$ is closed under homomorphic equivalence. However, 
\emph{unlike}  left query algorithms, 
a class of instances that admits a right query algorithm over 
$\bool$ is not necessarily FO-definable. Concretely, as we saw in Example \ref{exam:csp}, the class of 3-colorable graphs admits a  right query algorithm over $\bool$, but is not FO-definable.  In fact, 
\emph{every} constraint satisfaction problem $\csp(A)$ admits a right query algorithm over $\bool$, and the FO-definable ones are precisely those that admit a left query algorithm over $\bool$. The next result is straightforward.

\begin{proposition}\label{thm:query_alg_bool_right}
Let $\class{C}$ be a class of instances. Then $\class{C}$ admits a right query algorithm over $\bool$ if and only if $\class{C}$ is definable by a Boolean combination of CSPs.
\end{proposition}

We also saw in Section~\ref{sec:left_and_right_alg_bool}
that not every homomorphic-equivalence closed FO-definable class
admits a right query algorithm over $\bool$. For example, 
the class of triangle-free graphs does not admit a
right query algorithm  over $\bool$ (or even over $\nat$). This raises the question
which homomorphic-equivalence closed FO-definable classes admit a
right query algorithm over $\bool$. Equivalently, 
\emph{which classes $\class{C}$ that admit a 
left query algorithm over $\bool$,  admit  a right query algorithm over $\bool$}?
We will address this question next
by making use of two known results.

\begin{theorem}[Sparse Incomparability Lemma, \cite{kun2013constraints}]\label{thm:sparse}
Let $m, n \geq 0$. For every instance $B$ there is an instance $B^*$ of girth at least $m$ such that, for all instances $D$ with $\size{\adom(D)} \leq n$, we have $B \in \csp(D)$ if and only if $B^* \in \csp(D)$.
\end{theorem}

A \emph{finite homomorphism duality}  is a pair of
finite sets $(\set{F},\set{D})$ such that, 
for every instance $A$, the following are
equivalent: (i) $F\to A$ for some $F\in\set{F}$; (ii) $A\not\to D$ for all $D\in \set{D}$.  
We make use of the following
known characterization of this notion.

\begin{theorem}[\cite{foniok2008generalised}]\label{thm:foniok}
Let $A$ be an instance. Then the following statements are equivalent.
\begin{enumerate}
\item $A$ is homomorphically equivalent to an acyclic instance.
\item There is a finite homomorphism duality
$(\class{F},\class{D})$ with $\class{F}=\{A\}$.
\item There is a finite homomorphism duality
$(\class{F},\class{D})$ with $\class{F}=\{A\}$ and
where $\class{D}$ consists of instances $D$ for which
$\csp(D)$ is FO-definable.
\end{enumerate}
\end{theorem}

To see the implication
$(2) \Longrightarrow (3)$,  
note that, if $(\class{F},\class{D})$ is a finite
homomorphism duality, then $\bigcup_{D\in\class{D}}\csp(D)$ is a FO-definable class
(because it 
is defined by the negation of the UCQ $\bigvee_{F\in\class{F}}q^F$).
It follows, by the same reasoning as in the proof
of Proposition~\ref{prop:u-csp-left}, that exists
$\class{D}'\subseteq\class{D}$ such that
$(\class{F},\class{D}')$ is a finite homomorphism
duality and such that, for each $D\in \class{D}'$, we have that $\csp(D)$ is FO-definable.

Our next result characterizes the classes that
admit both a left query algorithm and a right query algorithm over $\bool$.

\begin{theorem}\label{thm:intersection-boolean}
    Let $\class{C}$ be a class of instances. Then the following statements are equivalent.
\begin{enumerate}
\item $\class{C}$ admits both a left query algorithm and a right query algorithm over $\bool$.
\item $\class{C}$ admits  a left query algorithm over $\nat$ and a right query algorithm over $\bool$.
\item $\class{C}$ is definable by a Boolean combination of Berge-acyclic CQs.
 \item $\class{C}$ is definable by a Boolean combination of FO-definable CSPs.
\end{enumerate}
\end{theorem}
\begin{proof}
The equivalence of statements (1) and (2) follows from Corollary~\ref{cor:lqa-nat-fo}.  The proof of the remaining equivalences is as follows.

\noindent $(1) \Longrightarrow (3)$ 
By Theorem~\ref{thm:char-left-alg-bool} and Proposition~\ref{thm:query_alg_bool_right},
$\class{C}$ is definable by a Boolean combination $\varphi$ of CQs, as well as by a Boolean combination
$\psi$ of CSPs. 
Let $\varphi'$  be obtained from $\varphi$ by replacing 
each conjunctive query $q$ by 
the disjunction of all homomorphic images of $q$ that are Berge-acyclic.
We claim that $\varphi'$ defines $\class{C}$.

By Theorem~\ref{thm:foniok}, 
we know that $\varphi'$ is also equivalent to a Boolean combination of CSPs,
which we may call $\psi'$.
Let $n$ be the maximum size of a CSP occurring in $\psi$ and $\psi'$.
Also, let $m$ be greater than the maximum size of the CQs in $\varphi$.
Let $B$ be any instance, and let $B^*$ now be the instance given by Theorem~\ref{thm:sparse} (for $m,n$ as chosen above).
By construction, $B$ and $B^*$ agree with each other on their membership in CSPs of size at most $n$, which
include all CSPs occurring in $\psi$ as well as
$\psi'$, and therefore, $B$ and $B^*$ agree on
$\psi$ and $\psi'$. Since $\psi$ is equivalent to $\varphi$ and $\psi'$ is equivalent to $\varphi'$, this means
that $B$ and $B^*$ agree on $\varphi$ and $\varphi'$.
Furthermore, by construction, $\varphi'$ is equivalent to $\varphi$ on 
instances of girth at least $m$ (because every homomorphic image
of a CQ of size less than $m$ in such an instance must be acyclic). In particular, 
It follows that $B^*\models\varphi$ iff $B^*\models\varphi'$. 
Putting everything together, we have that
$B\models\varphi$ iff $B^*\models\varphi$
iff $B^*\models\varphi'$ iff $B\models\varphi'$.

\noindent $(3) \Longrightarrow (4)$ This follows from Theorem~\ref{thm:foniok} (for $A$ the canonical instance of $q$): we simply replace
each Berge-acyclic conjunctive query $q$ by the conjunction $\bigwedge_{D\in\class{D}} \neg \csp(D)$.

\noindent $(4) \Longrightarrow (1)$ This follows from  Theorem~\ref{thm:char-left-alg-bool} and Proposition \ref{thm:query_alg_bool_right}.
\end{proof}

Let $\class{C}$ be a  class that admits a
left query algorithm over $\bool$ or,  equivalently,  let $\class{C}$ be definable by a Boolean combination of  CQs. It follows  that $\class{C}$ admits a right query algorithm over $\bool$
if and only if $\class{C}$ is definable by a Boolean combination of \emph{Berge-acyclic} CQs.
Similarly, let $\class{C}$ be a class that admits 
a right query algorithm over $\bool$ or, equivalently,  let $\class{C}$ be definable by a Boolean combination of CSPs. It follows that $\class{C}$ admits a left query algorithm over $\bool$
if and only if $\class{C}$ is definable by a Boolean 
combination of \emph{FO-definable} CSPs.

Finally, we will consider the question
when a class of the form $\homtype{A}$ admits
a right query algorithm over $\bool$.

\begin{restatable}{theorem}{thmhomequivclassright}\label{thm:hom-equiv-class-right}
Let $A$ be an instance. Then the following statements are equivalent.
\begin{enumerate}
\item $\homtype{A}$ has a right query algorithm over $\bool$.
\item $\setm{B}{A \homto B}$ has a right query algorithm over $\bool$.
    \item $A$ is homomorphically equivalent to an acyclic instance.
\end{enumerate}
Moreover, testing whether this holds (for a given instance $A$)
is NP-complete. 
\end{restatable}

\begin{proof} {}
We will close a cycle of implications.

\noindent $(1) \Longrightarrow (2)$:
For this, we  use the exponentiation operation $X^Y$~\cite{HellNesetril2004}. This
operation has the property that $X\to Y^Z$ if and only if  $X \dprod Z\to Y$.
Assume $\setm{B}{A \homto B}$ does not admit a right query algorithm over $\bool$. Consider an arbitrary finite set of instances $\set{F} = \sete{F_1, \etl, F_k}$ (think: candidate right query algorithm for $\homtype{A}$). Since $\setm{B}{A \homto B}$ does not admit a right query algorithm over $\bool$, for the set $\set{F}^{A} \defas \sete{F_1^{A}, \etl, F_k^{A}}$ there are instances $P$ and $Q$ with $A \homto P$ and $A \not\homto Q$ such that $\hom_\bool(P, \set{F}^{A}) = \hom_\bool(Q, \set{F}^{A})$ or, equivalently, $\hom_\bool(P \dprod A, \set{F}) = \hom_\bool(Q \dprod A, \set{F})$. Let $P' \defas P \dprod A$ and let $Q' \defas Q \dprod A$. Then, by Proposition~\ref{prop:hom-combinatorics}, $P' \in \homtype{A}$ and $Q' \notin \homtype{A}$ but $\hom_\bool(P', \set{F}) = \hom_\bool(Q', \set{F})$. Therefore $\homtype{A}$ has no right query algorithm over $\bool$.

\noindent $(2) \Longrightarrow (3)$:  Assume that $\{B\mid A\to B\}$ admits a right query algorithm for $\bool$.
We will construct a finite homomorphism duality
$(\class{F},\class{D})$ with $\class{F}=\{A\}$.
It then follows by Theorem~\ref{thm:foniok} that 
$A$ is homomorphically equivalent to an acyclic instance.
Let $B_1, \ldots B_n$ be all those instances $B_i$ used by the right query algorithm
   for which it holds that $A\not\to B_i$. We claim that 
   $(\{A\},\{B_1, \ldots, B_n\})$ is a finite homomorphism duality.
   Let $C$ be any instance. If $C\to B_i$ for some $i\leq n$, then $A\not\to C$ (otherwise, by transitivity, we would have that $A\to B_i$). Conversely, if
   $A\not\to C$, then the algorithm
   must answer ``no'' on input $C$ while it answers ``yes'' on
   input $C\dsum A$. Therefore,
   one of the right-queries made by the algorithm
   must differentiate $C$ from $C\dsum A$. It is easy to see that the right-query in 
   question must consist of an instance into which $C$ maps but $A$
   does not. This instance must then be among the $B_i$, and 
    $C\to B_i$.

\noindent $(3)\Longrightarrow (1)$:  By Theorem~\ref{thm:foniok}, 
there is a finite homomorphism duality
$(\{A\}, \class{D})$. In particular,
 for all instances $C$,
we have that $C\in \homtype{A}$ if and only if $C\to A$ and 
$C\not\to D$ for all $D\in\class{D}$.

To test if a given instance is homomorphically equivalent
to an acyclic instance, it suffices to test that its
core is acyclic (equivalently, that it has an
acyclic retract). This can clearly be done in NP. The NP-hardness follows directly from
Theorem 6 in \cite{DalmauKV02}.
\end{proof}
The preceding Theorem \ref{thm:hom-equiv-class-right} can be thought of as an  analogue of Proposition~\ref{prop:positive-homtype-csp} 
for right query algorithms.
Again, this result extends to finite unions of homomorphic-equivalence classes.

\begin{restatable}{theorem}{thmhomequivclassrightunion}
\label{thm:hom-equiv-class-right-union}
For all instances $A_1, \etl, A_n$, the following statements are equivalent.
\begin{enumerate}
\item $\bunion_{1 \leq i \leq n} \homtype{A_i}$
admits a right query algorithm over $\bool$.
\item Each $\homtype{A_i}$, for $i = 1, \etl, n$, admits a right query algorithm over $\bool$.
\end{enumerate}
In particular, testing whether this holds (for given instances $A_1, \ldots, A_n$) is NP-complete.
\end{restatable}

\begin{proof} {} 
It is clear that the second statement  implies the first. We will prove by induction on $n$ that the first statement implies the second. The base case ($n = 1$) is immediate since then statements  (1) and (2) coincide. Next, let $n > 1$ and $\class{C} \defas \bunion_{1 \leq i \leq n} \homtype{A_i}$. We proceed by contraposition, assuming that $\homtype{A_i}$ does not admit a right query algorithm over $\bool$ for some $i \leq n$. We may assume without loss of generality that $A_1, \ldots, A_n$ are pairwise not homomorphically equivalent. Note that $\homto$ induces a preorder among $A_1, \etl, A_n$ and, since $n$ is finite, there is a minimal. Without loss of generality, assume that $A_n$ is a minimal, that is, $A_i \not\homto A_n$ for all $i < n$. We distinguish two cases.

(1) $\homtype{A_n}$ admits a right query algorithm over $\bool$. Then, for some $i \leq n - 1$, $\homtype{A_i}$ does not admit any right query algorithm over $\bool$. By induction hypothesis, we have $\class{C}' \defas \bunion_{1 \leq i \leq n - 1} \homtype{A_i}$ does not admit any right query algorithm over $\bool$. Then $\class{C}$ does not admit a right query algorithm over $\bool$ either, since $\class{C}' = \class{C} \setminus \homtype{A_n}$.

(2) $\homtype{A_n}$ does not admit any right query algorithm over $\bool$. By Theorem~\ref{thm:hom-equiv-class-right}, the class $\setm{B}{A \homto B}$ does not admit any right query algorithm over $\bool$, either. Consider an arbitrary finite non-empty set of instances $\set{F} = \sete{F_1, \etl, F_k}$. Since $\setm{B}{A \homto B}$ does not admit any right query algorithm over $\bool$, for the set $\set{F}^{A_n} = \sete{F_1^{A_n}, \etl, F_k^{A_n}}$ there are instances $P$ and $Q$ with $A_n \homto P$ and $A_n \not\homto Q$ such that $\hom_\bool(P, \set{F}^{A_n}) = \hom_\bool(Q, \set{F}^{A_n})$, which implies that $\hom_\bool(P \dprod A_n, \set{F}) = \hom_\bool(Q \dprod A_n, \set{F})$. It follows by Proposition~\ref{prop:hom-combinatorics} that
\begin{itemize}
\item $(P \dprod A_n) \in \homtype{A_n}$ because $A_n \homto P$,
\item $(Q \dprod A_n) \notin \homtype{A_n}$ because $A_n \not\homto Q$,
\item for all $i < n$, $(Q \dprod A_n) \notin \homtype{A_i}$ because $A_i \not\homto A_n$.
\end{itemize}
Let $P' \defas P \dprod A_n$ and let $Q' \defas Q \dprod A_n$. Then the above discussion yields that $P' \in \class{C}$ and $Q' \notin \class{C}$ while $\hom_\bool(P', \set{F}) = \hom_\bool(Q', \set{F})$. Therefore, $\class{C}$ does not admit any right query algorithm over $\bool$.
\end{proof}

\begin{remark}\label{rem:no-rqa}
We saw in Example~\ref{exam:triangle-free}
that the class of triangle-free graphs,
which clearly has a left query algorithm over $\bool$,
does
not admit a right query algorithm over $\bool$ or over $\nat$.
Observe that this class is defined by the
negation of the ``triangle'' conjunctive query  $\exists xyz (R(x,y)\land R(y,z)\land R(z,x))$. 
In light of Theorem~\ref{thm:intersection-boolean}, 
the lack of a right query algorithm over $\bool$ for this class
can be ``explained'' by the fact that this conjunctive query is
not Berge-acyclic.
Furthermore, in Example~\ref{exam:hom-equiv-K3} we
mentioned that the class $\homtype{K_3}$, that is,
the class of graphs that are 3-colorable and also 
contain a triangle, does not admit a right query algorithm over
$\bool$. This follows from Theorem~\ref{thm:hom-equiv-class-right}.
\end{remark}

We conclude this section with an open problem.

\begin{question}
Does a suitable analogue of Theorem~\ref{thm:lqa-nat-fo1} hold for right query algorithms?
\end{question}

\section{Summary and Discussion of Related Topics}\label{sec:conclus}

Inspired by the work of Chen et al.~\cite{chen2022algorithms}, 
we extended their framework and studied various types of
query algorithms, where a query algorithm for a class $\class{C}$ of instances determines whether a given input instance belongs to $\class{C}$  by making a finite number of (predetermined) queries that  ask for the existence of
certain homomorphisms or for the number of certain homomorphisms.
Specifically, we introduced
and studied \emph{left query algorithms} and \emph{right query
algorithms} over $\bool$, as well as   over $\nat$. 
Our results delineate the differences in expressive power between these four types of 
query algorithms. In particular, they pinpoint when
the ability to count homomorphisms is essential for 
the existence of left query algorithms.

\smallskip

\noindent{\bf Relationship to view determinacy}~
Recently, Kwiecien et al.~\cite{Kwiecien2022:determinacy}
studied view determinacy under  bag semantics.
In particular, they obtained a decidability result
for determinacy with respect to 
Boolean views under
 bag-set semantics.
We will briefly describe their framework and relate it to ours. At the most abstract level, a 
\emph{view} is simply a function $f$ that
takes as input a database instance. Specifically, 
under  set semantics, every Boolean CQ specifies
a view that is a function from 
database instances to $\{0,1\}$, 
while, under  bag-set semantics, 
every Boolean CQ specifies a view that is a function from
database instances to $\nat$. 
We say that a finite collection of views $f_1, \ldots, f_k$
\emph{determines} a view $g$, if
for all database instances $A$ and $B$,
if $f_i(A)=f_i(B)$ for all $i\leq k$, then $g(A)=g(B)$. 
The aforementioned result from~\cite{Kwiecien2022:determinacy} asserts
that the following problem is decidable:
given views $f_1, \ldots, f_k$ and $g$
specified by Boolean CQs under  bag-set
semantics, is
$g$ is determined by $f_1, \ldots, f_k$? 

We now describe the relationship between the above notion of view determinacy and our framework.
Let $\class{C}$ be a class of instances and let $\class{F}=\{F_1, \ldots, F_k\}$ be
a finite set of  instances. 
Then the following are equivalent:
\begin{enumerate}
    \item There exists a set $X$ such that 
    $(\class{F},X)$ is a left query algorithm over $\nat$ for $\class{C}$,
    \item $f_1, \ldots, f_k$ determine $g_{\class{C}}$
where $f_i(A) = \hom_\nat(F_i,A)$ and 
$g_{\class{C}}$ is the indicator function of $\class{C}$
(i.e., $g_{\class{C}}(A)=1$ if $A\in\class{C}$ and  $g_{\class{C}}(A)=0$ otherwise).
\end{enumerate}
This tells that there are some important differences between our
framework and the one in \cite{Kwiecien2022:determinacy}:
(i) when we study the existence of left query algorithms,
the set $\class{F}$ is not fixed, whereas,
in the view determinacy problem, the views are given as 
part of the input; (ii) in the view determinacy problem
studied in \cite{Kwiecien2022:determinacy}, the view $g$
is specified by a CQ with bag-set semantics, whereas
in our case $g$ is a Boolean-valued function since it is
the indicator function of a class of instances,
(iii) we do not assume that the
class $\class{C}$ is specified by a Boolean CQ. Indeed, if
$\class{C}$ were specified by a Boolean CQ $q$, then a left query 
algorithm would trivially exist, where $\class{F}$ simply
consists of (the canonical instance of) $q$ itself.

\smallskip

\noindent{\bf Other semirings}~
Query algorithm over $\bool$
and query algorithm over $\nat$ can be viewed as special cases
of a more general setting, namely that of a query algorithm over a semiring. There is a body of research in the interface of databases and semirings, including the study of provenance of database queries using semirings \cite{DBLP:conf/pods/GreenKT07,DBLP:journals/sigmod/KarvounarakisG12}, the study of the query containment problem under  semiring semantics \cite{DBLP:journals/mst/Green11,DBLP:journals/tods/KostylevRS14}, and, more recently, the study of Datalog under semiring semantics \cite{DBLP:conf/pods/Khamis0PSW22}. In these studies, the semirings considered are \emph{positive}, which means that they are commutative semirings with no zero divisors and with the property that the sum of any two non-zero elements is  non-zero.
It is perfectly meaningful to define homomorphism counts over positive semirings and then investigate query algorithms over such semirings. In particular, it may be interesting to investigate query algorithms over the \emph{tropical semiring} $\mathbb R = (R\cup \{\infty\}, \min,+)$, where $R$ is the set of real numbers, since it is well known that various \emph{shortest-distance} problems can be naturally captured using this semiring.

\smallskip

\noindent{\bf Adaptive query algorithms}~ 
The query algorithms $(\class{F},X)$ studied in this paper
are \emph{non-adaptive}, in the sense that the set $\class{F}=\{F_1, \ldots, F_k\}$ is fixed up-front.
In contrast, an \emph{adaptive} query algorithm may 
decide the set of instances $\class{F}$ at run-time, that is to say, the choice of $F_i$ may depend
on the homomorphism-count vector for $F_1, \ldots, F_{i-1}$. 
As pointed out in the Introduction, whenever a class $\class{C}$
admits an adaptive left (right)  query algorithm over $\bool$, then it also admits a non-adaptive left (right) query algorithm over $\bool$. Note that the most ``economical'' (as regards the number of instances used) non-adaptive algorithm for a class 
$\class{C}$ may use a larger 
set $\set{F}$ of instances  than the adaptive one, but the number of instances used by the non-adaptive algorithm is still finite.
Thus, adaptive query algorithms over $\bool$ do not offer higher expressive power than adaptive ones. The situation for $\nat$ is 
quite different:  it was shown in \cite{chen2022algorithms} 
that \emph{every} isomorphism-closed class of instances 
(in particular, every homomorphic-equivalence closed class) 
admits
 an adaptive left $k$-query algorithm over $\nat$ already for $k=3$; therefore, adaptive left query algorithms over $\nat$ have higher expressive power than adaptive left query algorithms over $\bool$, even when it comes to 
 homomorphic-equivalence closed classes.
Switching sides, note that the class of triangle-free graphs (Example~\ref{exam:triangle-free}) does not admit
an adaptive right query algorithm over $\nat$,
as was shown in \cite[Proposition 8.2]{chen2022algorithms}; 
 hence it is a meaningful question to ask: which homomorphic-equivalence closed classes admit an adaptive right query algorithm over $\nat$?
\bibliography{main}
\end{document}